\documentclass[twoside]{article}

\usepackage[USenglish]{babel}
\usepackage{amsmath,amssymb,amsthm,dsfont,url,hyperref,cite,qtree,algorithmic}

\DeclareMathOperator{\Tr}{Tr}
\DeclareMathOperator{\openone}{\mathds{1}}
\DeclareMathOperator{\brn}{\mathtt{branch}}
\DeclareMathOperator{\bnd}{\mathtt{bound}}
\DeclareMathOperator*{\argmax}{arg\,max}

\newtheorem{dfn}{Definition}
\newtheorem{lmm}{Lemma}
\newtheorem{thm}{Theorem}
\newtheorem{cor}{Corollary}

\begin{document}
\setlength{\textheight}{8.0truein}


\thispagestyle{empty}
\setcounter{page}{1}


\vspace*{0.88truein}



\centerline{\bf Adversarial guesswork with quantum side information}
\vspace*{0.035truein}

\centerline{\footnotesize Baasanchimed Avirmed}

\vspace*{0.015truein}

\centerline{\footnotesize\it Department  of Computer Science
  and Engineering, Toyohashi University of Technology}

\baselineskip=10pt

\centerline{\footnotesize\it  1-1 Hibarigaoka,  Tempaku-cho,
  Toyohashi, Aichi, 441-8580, Japan}

\vspace*{10pt}

\centerline{\footnotesize Kaito Niinomi}

\vspace*{0.015truein}

\centerline{\footnotesize\it Department  of Computer Science
  and Engineering, Toyohashi University of Technology}

\baselineskip=10pt

\centerline{\footnotesize\it  1-1 Hibarigaoka,  Tempaku-cho,
  Toyohashi, Aichi, 441-8580, Japan}

\vspace*{10pt}

\centerline{\footnotesize Michele Dall'Arno}

\vspace*{0.015truein}

\centerline{\footnotesize\it Department  of Computer Science
  and Engineering, Toyohashi University of Technology}

\baselineskip=10pt

\centerline{\footnotesize\it  1-1 Hibarigaoka,  Tempaku-cho,
  Toyohashi, Aichi, 441-8580, Japan}

\baselineskip=10pt

\centerline{\footnotesize\it michele.dallarno.mv@tut.jp}

\vspace*{0.225truein}


\vspace*{0.21truein}

\begin{abstract}
  The  guesswork of  a classical-quantum  channel quantifies
  the cost incurred in guessing the state transmitted by the
  channel  when only  one state  can be  queried at  a time,
  maximized over any  classical pre-processing and minimized
  over      any      quantum      post-processing.       For
  arbitrary-dimensional      covariant     classical-quantum
  channels,  we   prove  the   invariance  of   the  optimal
  pre-processing   and  the   covariance   of  the   optimal
  post-processing.   In  the  qubit  case,  we  compute  the
  optimal  guesswork  for  the  class  of  so-called  highly
  symmetric   informationally   complete   classical-quantum
  channels.
\end{abstract}

\vspace*{10pt}

\vspace*{3pt}


\section{Introduction}

Let us first introduce an adversarial extension of the usual
guesswork    problem     in    the    absence     of    side
information~\cite{Mas94, Ari96,  AM98a, AM98b,  MS04, Sun07,
  HS10, CD13, SV18,  Sas18}.  One party, say  Alice, is free
to  choose  a  probability  distribution $p$  for  a  random
variable $M$ over alphabet  $\mathcal{M}$, and to communicate
her choice  to the other  party, say Bob (in  the previously
considered, non-adversarial  scenario, $p$  is fixed  by the
rules of the game).  At each  round of the game, Alice picks
a  value  $m$  for  variable  $M$  at  random  according  to
distribution $p$,  and Bob queries  Alice for the  values of
random  variable $M$,  one at  a  time, until  his guess  is
correct. For instance, let us consider the case $\mathcal{M}
= \{a, b,  c\}$.  In this case, Bob's first  query could be,
say, $b$.  If  Alice answers on the negative,  then his next
query could be $a$.  Assuming this time Alice answers on the
affirmative, the round is over.

Bob chooses  the order of  his queries in order  to minimize
the cost  incurred, the  cost function  being known  to both
parties in advance and only  depending on the average number
of queries; Alice chooses the prior probability distribution
$p$ to  maximize such  a cost.   The optimal  strategies for
both Alice  and Bob  are obvious: for  Alice it  consists of
choosing $p$ as the uniform distribution over $\mathcal{M}$,
while for Bob  it consists of querying the values  of $M$ in
non-increasing order of their prior probability.

Let  us  now  introduce  the adversarial  extension  of  the
guesswork    problem    in     the    presence    of    side
information~\cite{CCWF15,  HKDW22,   DBK22,  DBK21,  Dal23},
that, most  generally, is quantum (again,  the quantum cases
so far considered  were not adversarial, that  is, the prior
$p$ was  assumed fixed by the  rules).  That is, let  us now
say    that    there     is    a    communication    channel
$\boldsymbol{\sigma}$ with random variable  $M$ as input and
quantum  states  as  output (a  classical-quantum,  or  c-q,
channel  for short),  known to  both parties.   Suppose also
that,  at each  round  of the  game, Bob  is  given a  state
$\boldsymbol{\sigma} ( m )$.

How would the optimal strategies  for both parties look like
in  this case?   Bob is  free  to perform  the most  general
quantum measurement on his state in order to get a posterior
probability distribution on $M$,  and queries Alice based on
such  a posterior  (it was  shown in  Ref.~\cite{Dal23} that
more general  strategies that  make use of  Alice's feedback
after each  query do not  help); Alice chooses the  prior in
order to antagonize such an optimal strategy.  In this case,
therefore,  the optimal  strategies  on both  sides are  not
obvious   at   all   and   depend   on   the   c-q   channel
$\boldsymbol{\sigma}$.

In  this   work,  after  formalizing  the   problem  of  the
adversarial  guesswork  with  quantum  side  information  in
Section~\ref{sec:formalization},  we  address the  arbitrary
dimensional  case  (Section~\ref{sec:arbitrary}).  We  prove
that  the  order  in  which   Alice  and  Bob  choose  their
strategies  in  irrelevant  (Lemma~\ref{lmm:minimax});  that
each  choice  of  strategy   amounts  to  a  convex  problem
(Lemma~\ref{lmm:convexity});    that    any   symmetry    of
$\boldsymbol{\sigma}$, if any, implies an analog symmetry of
Alice's       and       Bob's       optimal       strategies
(Lemma~\ref{lmm:covariance});   and   that   Bob's   optimal
strategy amounts to querying the values of $M$ in decreasing
order   of   their    posterior   probability   distribution
(Lemma~\ref{lmm:bayes}).

Then,     we      specify     to     the      qubit     case
(Section~\ref{sec:qubit}).    Using  the   fact,  shown   in
Ref.~\cite{Dal23}, that  the optimization of  Bob's strategy
can  be  reframed  as   a  combinatorial  problem  known  as
quadratic assignment~\cite{KB57, SG76,  LM88, Cel98, BCPP98,
  BCRW98}, we derive a  branch and bound~\cite{LD60, LMSK63,
  Jen99} (BB)  algorithm for the closed-form  computation of
the  guesswork  (Theorem~\ref{thm:bb}).   We apply  such  an
algorithm  to  compute  the closed-form  expression  of  the
guesswork of  the highly-symmetric, informationally-complete
(HSIC)   c-q   channels   introduced   in   Ref.~\cite{SS16}
(Corollary~\ref{cor:hsic}), and we provide an implementation
with  the  IBM quantum  computer  of  the optimal  guesswork
protocol  for   the  icosidodecahedral  HSIC   channel.   We
summarize  our results  and  discuss some  open problems  in
Section~\ref{sec:conclusion}.

\section{Formalization}
\label{sec:formalization}

We  use   standard  definitions   and  results   in  quantum
information theory~\cite{Wil17}.

For  any finite-dimensional  Hilbert space  $\mathcal{H}$ we
denote  with  $\mathcal{L}_+   (\mathcal{H})$  the  cone  of
positive semidefinite  operators on $\mathcal{H}$.

For  any   finite  set  $\mathcal{M}$  we   define  the  set
$\mathcal{D}_\mathcal{M}$ of  probability distributions over
$\mathcal{M}$ given by
\begin{align*}
  \mathcal{D}_{\mathcal{M}} :=  \left\{ p :  \mathcal{M} \to
  \left[  0,  1 \right]  \Big|  \sum_{m  \in \mathcal{M}}  p
  \left( m \right) = 1 \right\},
\end{align*}
the  set $\mathcal{C}  ( \mathcal{M},  \mathcal{H})$ of  c-q
channels given by
\begin{align*}
  \mathcal{C}  \left( \mathcal{M},  \mathcal{H} \right)  : =
  \left\{    \boldsymbol{\sigma}     :    \mathcal{M}    \to
  \mathcal{L}_+ \left( \mathcal{H} \right) \Big | \Tr \left[
    \boldsymbol{\sigma}  \left( \cdot  \right)  \right] =  1
  \right\},
\end{align*}
the   set  $\mathcal{N}_{\mathcal{M}}$   of  numberings   of
$\mathcal{M}$ given by
\begin{align*}
  \mathcal{N}_{\mathcal{M}} := \left\{  \mathbf{n} : \left\{
  1,   \dots,  \left|   \mathcal{M}  \right|   \right\}  \to
  \mathcal{M} \Big| \mathbf{n} \textrm{ bijective} \right\},
\end{align*}
and  the   set  $\mathcal{P}   (  \mathcal{N}_{\mathcal{M}},
\mathcal{H} )$ of numbering-valued measurements given by
\begin{align*}
  \mathcal{P} \left(  \mathcal{N}_{\mathcal{M}}, \mathcal{H}
  \right)      :=       \left\{      \boldsymbol{\pi}      :
  \mathcal{N}_{\mathcal{M}}    \to   \mathcal{L}_+    \left(
  \mathcal{H}    \right)    \Big|    \sum_{\mathbf{n}    \in
    \mathcal{N}_{\mathcal{M}}}    \boldsymbol{\pi}    \left(
  \mathbf{n} \right) = \openone \right\}.
\end{align*}

For  any finite  set $\mathcal{M}$,  any finite  dimensional
Hilbert    space    $\mathcal{H}$,     any    c-q    channel
$\boldsymbol{\sigma}    \in     \mathcal{C}    (\mathcal{M},
\mathcal{H})$,   and    any   numbering-valued   measurement
$\boldsymbol{\pi}         \in          \mathcal{P}         (
\mathcal{N}_{\mathcal{M}},  \mathcal{H})$,  we  denote  with
$p_{\boldsymbol{\sigma}, \boldsymbol{\pi}}$  the probability
distribution  that  the  outcome  of  $\boldsymbol{\pi}$  is
$\mathbf{n}$ and the $t$-th query is correct, that is
\begin{align*}
  &      p_{p, \boldsymbol{\sigma},      \boldsymbol{\pi}}      :
  \mathcal{N}_{\mathcal{M}}  \times  \left\{  1 ,  \dots,  |
  \mathcal{M} | \right\} \to [0, 1]\\ & \left( \mathbf{n}, t
  \right)  \mapsto  p  \left( \mathbf{n}  \left(  t  \right)
  \right)  \Tr  \left[  \boldsymbol{\pi}  \left(  \mathbf{n}
    \right) \boldsymbol{\sigma}  \left( \mathbf{n}  \left( t
    \right) \right) \right],
\end{align*}
for any  $\mathbf{n} \in \mathcal{N}_{\mathcal{M}}$  and any
$t \in  \{ 1 ,  \dots, | \mathcal{M}  | \}$. We  denote with
$q_{\boldsymbol{\sigma}, \boldsymbol{\pi}}$  the probability
distribution  that the  $t$-th  guess  is correct,  obtained
marginalizing  $p_{\boldsymbol{\sigma},  \boldsymbol{\pi}}$,
that is
\begin{align*}
  & q_{p,  \boldsymbol{\sigma}, \boldsymbol{\pi}}  : \left\{
  1,  \dots,   |  \mathcal{M}  |  \right\}   \to  \left[  0,
    1\right]\\    &   t    \mapsto   \sum_{\mathbf{n}    \in
    \mathcal{N}_{\mathcal{M}}}       p_{\boldsymbol{\sigma},
    \boldsymbol{\pi}} \left( \mathbf{n}, t \right),
\end{align*}

For any cost function $\gamma : \{ 1, \dots, | \mathcal{M} |
\}$,  the  guesswork $G^\gamma  :  \mathcal{D}_{\mathcal{M}}
\times   \mathcal{C}   (\mathcal{M},   \mathcal{H})   \times
\mathcal{P}  (  \mathcal{N}_{\mathcal{M}}, \mathcal{H})  \to
\mathbb{R}$ is given by
\begin{align*}
  &     G^{\gamma}     \left(    p,     \boldsymbol{\sigma},
  \boldsymbol{\pi}  \right) \\  :=  & \sum_{\substack{t  \in
      \left\{   1,   \dots,   \left|   \mathcal{M}   \right|
      \right\}\\ \mathbf{n}  \in \mathcal{N}_{\mathcal{M}}}}
  p \left(  \mathbf{n} \left(  t \right) \right)  \Tr \left[
    \boldsymbol{\sigma} \left(  \mathbf{n} \left(  t \right)
    \right)   \boldsymbol{\pi}  \left(   \mathbf{n}  \right)
    \right] \gamma \left( t \right),
\end{align*}
for     any     probability      distribution     $p     \in
\mathcal{D}_{\mathcal{M}}$,       any      c-q       channel
$\boldsymbol{\sigma}    \in     \mathcal{C}    (\mathcal{M},
\mathcal{H})$,   and    any   numbering-valued   measurement
$\boldsymbol{\pi}         \in          \mathcal{P}         (
\mathcal{N}_{\mathcal{M}}, \mathcal{H})$.

The      minimum     guesswork      $G_{\min}^{\gamma}     :
\mathcal{D}_{\mathcal{M}}  \times \mathcal{C}  (\mathcal{M},
\mathcal{H}) \to \mathbb{R}$ is given by
\begin{align*}
  G_{\min}^{\gamma} \left( p, \boldsymbol{\sigma} \right) :=
  \min_{\boldsymbol{\pi}     \in      \mathcal{P}     \left(
    \mathcal{N}_{\mathcal{M}},      \mathcal{H}     \right)}
  G^{\gamma} \left( p, \boldsymbol{\sigma}, \boldsymbol{\pi}
  \right),
  \end{align*}
for     any     probability      distribution     $p     \in
\mathcal{D}_{\mathcal{M}}$     and    any     c-q    channel
$\boldsymbol{\sigma}    \in     \mathcal{C}    (\mathcal{M},
\mathcal{H})$,      and      the      maximin      guesswork
$G_{\max\min}^{\gamma}    :     \mathcal{C}    (\mathcal{M},
\mathcal{H}) \to \mathbb{R}$ is given by
\begin{align}
  \label{eq:gwmaximin}
  G_{\max\min}^{\gamma} \left( \boldsymbol{\sigma} \right)
  := \max_{p \in \mathcal{D}_{\mathcal{M}}}
  G_{\min}^{\gamma} \left( p, \boldsymbol{\sigma} \right),
  \end{align}
for  any c-q  channel  $\boldsymbol{\sigma} \in  \mathcal{C}
(\mathcal{M}, \mathcal{H})$.

The  maximum guesswork  $G_{\max}^{\gamma}  : \mathcal{C}  (
\mathcal{M},    \mathcal{H}   )    \times   \mathcal{P}    (
\mathcal{N}_{\mathcal{M}},  \mathcal{H}) \to  \mathbb{R}$ is
given by
\begin{align*}
  G_{\max}^{\gamma}        \left(       \boldsymbol{\sigma},
  \boldsymbol{\pi}      \right)      :=     \max_{p      \in
    \mathcal{D}_{\mathcal{M}}}    G^{\gamma}    \left(    p,
  \boldsymbol{\sigma}, \boldsymbol{\pi} \right),
\end{align*}
for  any c-q  channel  $\boldsymbol{\sigma} \in  \mathcal{C}
(\mathcal{M},   \mathcal{H})$   and   any   numbering-valued
measurement      $\boldsymbol{\pi}      \in      \mathcal{P}
(\mathcal{N}_{\mathcal{M}},  \mathcal{H})$, and  the minimax
guesswork $G_{\min\max}^{\gamma} : \mathcal{C} (\mathcal{M},
\mathcal{H}) \to \mathbb{R}$ given by
\begin{align}
  \label{eq:gwminimax}
  G_{\min\max}^{\gamma}  \left( \boldsymbol{\sigma}  \right)
  :=    \min_{\boldsymbol{\pi}   \in    \mathcal{P}   \left(
    \mathcal{N}_{\mathcal{M}},      \mathcal{H}     \right)}
  G_{\max}^{\gamma}        \left(       \boldsymbol{\sigma},
  \boldsymbol{\pi} \right),
\end{align}
for  any c-q  channel  $\boldsymbol{\sigma} \in  \mathcal{C}
(\mathcal{M}, \mathcal{H})$.

\section{Main results}

\subsection{Arbitrary dimensional case}
\label{sec:arbitrary}

The following lemma shows  that, without loss of generality,
we     can     focus     on    the     maximin     guesswork
$G_{\max\min}^{\gamma}$ only.

\begin{lmm}[Maximin]
  \label{lmm:minimax}
  For any  finite set $\mathcal{M}$,  any finite-dimensional
  Hilbert    space    $\mathcal{H}$,   any    c-q    channel
  $\boldsymbol{\sigma}   \in   \mathcal{C}  (   \mathcal{M},
  \mathcal{H} )$, and any function $\gamma : \{ 1, \dots , |
  \mathcal{M}     |    \}$,     the    maximin     guesswork
  $G_{\max\min}^{\gamma}  ( \boldsymbol{\sigma}  )$ and  the
  minimax      guesswork       $G_{\min\max}^{\gamma}      (
  \boldsymbol{\sigma} )$ are equivalent, that is
  \begin{align*}
    G_{\max\min}^{\gamma} \left( \boldsymbol{\sigma} \right)
    =   G_{\min\max}^{\gamma}   \left(   \boldsymbol{\sigma}
    \right)
  \end{align*}
\end{lmm}

\begin{proof}
  The  statement  immediately  follows  from  von  Neumann's
  minimax    theorem    by    observing   that    the    set
  $\mathcal{D}_{\mathcal{M}}$  of probability  distributions
  and  the  set  $\mathcal{P}  (  \mathcal{N}_{\mathcal{M}},
  \mathcal{H}  )$   of  numbering-valued   measurements  are
  compact  and  that  the  guesswork  $G^{\gamma}  (  \cdot,
  \boldsymbol{\rho}, \cdot )$ is bilinear over such sets.
\end{proof}

Next,  we prove  the  convexity and  the  covariance of  the
maximization problem  over probability distributions  in the
right  hand  side  of Eq.~\eqref{eq:gwmaximin}  and  of  the
minimization problem  over numbering-valued  measurements in
the  right hand  side of  Eq.~\eqref{eq:gwminimax}.  Despite
sharing  these  properties, we  will  show  that the  former
problem is ``easy'' while the latter problem is ``hard'', in
terms of finding a closed-form  solution as well as in terms
of complexity class.

The following lemma shows  that the maximization problem in
the  right hand  side  of  Eq.~\eqref{eq:gwmaximin} and  the
minimization   problem   in   the  right   hand   side   of
Eq.~\eqref{eq:gwminimax} are convex programming problems.

\begin{lmm}[Convexity]
  \label{lmm:convexity}
  For any  finite set $\mathcal{M}$,  any finite-dimensional
  Hilbert    space    $\mathcal{H}$,   any    c-q    channel
  $\boldsymbol{\sigma}    \in   \mathcal{C}    (\mathcal{M},
  \mathcal{H})$, and any function $\gamma  : \{ 1, \dots , |
  \mathcal{M}     |    \}$,     the    minimum     guesswork
  $G_{\min}^{\gamma}  ( \cdot,  \boldsymbol{\sigma} )$  is a
  concave function over  the set $\mathcal{D}_{\mathcal{M}}$
  of  probability distributions  and  the maximum  guesswork
  $G_{\max}^{\gamma}  ( \boldsymbol{\sigma},  \cdot )$  is a
  convex    function   over    the   set    $\mathcal{P}   (
  \mathcal{N}_{\mathcal{M}},      \mathcal{H}     )$      of
  numbering-valued measurements.
\end{lmm}

\begin{proof}
  The first part of the statement immediately follows by the
  linearity  of $G^{\gamma}$  over  the  set $\mathcal{P}  (
  \mathcal{N}_{\mathcal{M}},      \mathcal{H}     )$      of
  numbering-valued  measurements.    Indeed,  for   any  two
  probability       distributions       $p,      q       \in
  \mathcal{D}_{\mathcal{M}}$ and  any probabilities $\lambda
  \in [0, 1]$ and $\mu := 1 - \lambda$ one has
  \begin{align*}
    &   G_{\min}^{\gamma}  \left(   \lambda  p   +  \mu   q,
    \boldsymbol{\sigma} \right)\\ = & \min_{\boldsymbol{\pi}
      \in   \mathcal{P}  \left(   \mathcal{N}_{\mathcal{M}},
      \mathcal{H} \right)} G^{\gamma} \left( \lambda p + \mu
    q, \boldsymbol{\sigma},  \boldsymbol{\pi} \right)\\  = &
    \min_{\boldsymbol{\pi}     \in    \mathcal{P}     \left(
      \mathcal{N}_{\mathcal{M}}, \mathcal{H} \right)} \left[
      \lambda  G^{\gamma}   \left(  p,  \boldsymbol{\sigma},
      \boldsymbol{\pi}  \right) +  \mu G^{\gamma}  \left( q,
      \boldsymbol{\sigma},      \boldsymbol{\pi}     \right)
      \right]\\  \ge  & \lambda  \min_{\boldsymbol{\pi}  \in
      \mathcal{P}      \left(     \mathcal{N}_{\mathcal{M}},
      \mathcal{H}    \right)}     G^{\gamma}    \left(    p,
    \boldsymbol{\sigma},  \boldsymbol{\pi}   \right)  +  \mu
    \min_{\boldsymbol{\pi}     \in    \mathcal{P}     \left(
      \mathcal{N}_{\mathcal{M}},     \mathcal{H}    \right)}
    G^{\gamma}      \left(      q,      \boldsymbol{\sigma},
    \boldsymbol{\pi}    \right)     \\    =     &    \lambda
    G_{\min}^{\gamma} \left(  p, \boldsymbol{\sigma} \right)
    +  \mu G_{\min}^{\gamma}  \left( q,  \boldsymbol{\sigma}
    \right).
  \end{align*}

  The second  part of  the statement immediately  follows by
  the    linearity   of    $G^{\gamma}$    over   the    set
  $\mathcal{D}_{\mathcal{M}}$ of  probability distributions.
  Indeed,   for   any  two   numbering-valued   measurements
  $\boldsymbol{\pi},  \boldsymbol{\tau}  \in  \mathcal{P}  (
  \mathcal{N}_{\mathcal{M}},   \mathcal{H}    )$   and   any
  probabilities  $\lambda  \in  [0,  1]$ and  $\mu  :=  1  -
  \lambda$ one has
  \begin{align*}
    & G_{\max}^{\gamma}  \left( \boldsymbol{\sigma}, \lambda
    \boldsymbol{\pi} +  \mu \boldsymbol{\tau} \right)\\  = &
    \max_{p \in \mathcal{D}_{\mathcal{M}}} G^{\gamma} \left(
    p, \boldsymbol{\sigma},  \lambda \boldsymbol{\pi}  + \mu
    \boldsymbol{\tau}    \right)\\   =    &   \max_{p    \in
      \mathcal{D}_{\mathcal{M}}}  \left[ \lambda  G^{\gamma}
      \left(   p,    \boldsymbol{\sigma},   \boldsymbol{\pi}
      \right)     +     \mu     G^{\gamma}     \left(     p,
      \boldsymbol{\sigma},     \boldsymbol{\tau}     \right)
      \right]\\     \le     &    \lambda     \max_{p     \in
      \mathcal{D}_{\mathcal{M}}}   G^{\gamma}    \left(   p,
    \boldsymbol{\sigma},  \boldsymbol{\pi}   \right)  +  \mu
    \max_{p \in \mathcal{D}_{\mathcal{M}}} G^{\gamma} \left(
    p, \boldsymbol{\sigma}, \boldsymbol{\tau} \right) \\ = &
    \lambda  G_{\max}^{\gamma}  \left(  \boldsymbol{\sigma},
    \boldsymbol{\pi} \right) +  \mu G_{\max}^{\gamma} \left(
    \boldsymbol{\sigma}, \boldsymbol{\tau} \right).
  \end{align*}
\end{proof}

It is relevant to  compare the computational complexities of
the   two    problems   in   the   right    hand   side   of
Eq.~\eqref{eq:gwmaximin}  and~\eqref{eq:gwminimax}.  On  the
one hand, the size of the  problem in the right hand side of
Eq.~\eqref{eq:gwmaximin} grows linearly  with $| \mathcal{M}
|$; hence, its computational complexity  class is P.  On the
other hand, the  size of the problem in the  right hand side
of  Eq.~\eqref{eq:gwminimax}   grows  factorially   with  $|
\mathcal{M}  |$   (since  $|\mathcal{N}_{\mathcal{M}}   |  =
|\mathcal{M}|!$); even in the qubit case, such a problem has
been proven~\cite{Dal23} to be  a particular instance of the
quadratic assignment problem~\cite{Cel98} (QAP), a well-know
NP-hard combinatorial problem.

Let us turn  now to the symmetric case.  For  any finite set
$\mathcal{M}$,          the         symmetric          group
$\mathcal{S}_{\mathcal{M}}$ is given by
\begin{align*}
  \mathcal{S}_{\mathcal{M}} :=  \left\{ g :  \mathcal{M} \to
  \mathcal{M} \Big| g \textrm{ bijective} \right\}.
\end{align*}

For any  finite-dimensional Hilbert space  $\mathcal{H}$ and
any  c-q  channel  $\boldsymbol{\sigma}  \in  \mathcal{C}  (
\mathcal{M},  \mathcal{H}  )$,  we  say  that  a  map  $R  :
\mathcal{L}_+   (   \mathcal{H}   )  \to   \mathcal{L}_+   (
\mathcal{H})$  is  a  statistical  morphism~\cite{Bus12}  of
$\boldsymbol{\sigma}$ if  and only  if for any  discrete set
$\mathcal{N}$  and  any  measurement  $\boldsymbol{\pi}  \in
\mathcal{P}  (  \mathcal{N},  \mathcal{H} )$,  there  exists
measurement    $\boldsymbol{\tau}    \in    \mathcal{P}    (
\mathcal{N}, \mathcal{H} )$ such that
\begin{align}
  \label{eq:morphism}
  \Tr  \left[  \boldsymbol{\pi}  \left( n  \right)  R  \circ
    \boldsymbol{\sigma}  \left(  m  \right)  \right]  =  \Tr
  \left[     \boldsymbol{\tau}      \left(     n     \right)
    \boldsymbol{\sigma} \left( m \right) \right],
\end{align}
for any $m \in \mathcal{M}$ and any $n \in \mathcal{N}$.

For       any       group       $\mathcal{G}       \subseteq
\mathcal{S}_{\mathcal{M}}$,  we  say   that  a  c-q  channel
$\boldsymbol{\sigma}   \in    \mathcal{C}   (   \mathcal{M},
\mathcal{H})$  is  $\mathcal{G}$-covariant  if and  only  if
there  exists  a representation  $\mathcal{R}  :=  \{ R_g  :
\mathcal{L} ( \mathcal{H} )  \to \mathcal{L} ( \mathcal{H} )
\}$ of $\mathcal{G}$, where  $R_g$ is a statistical morphism
of $\boldsymbol{\sigma}$  for any $g \in  \mathcal{G}$, such
that
\begin{align}
  \label{eq:equivariance}
  R_g \circ \boldsymbol{\sigma}  = \boldsymbol{\sigma} \circ
  g.
\end{align}

For       any       group       $\mathcal{G}       \subseteq
\mathcal{S}_{\mathcal{M}}$  and any  $\mathcal{G}$-covariant
c-q channel $\boldsymbol{\sigma}$, we say that $\mathcal{G}$
is transitive iff there  exists one such representation such
that its action on  $\boldsymbol{\sigma}$ is transitive, and
we say that $\boldsymbol{\sigma}$ is centrally symmetric (CS
for short) iff there exists one such representation and a $g
\in \mathcal{G}$ such that $R_g( \cdot  ) = 2 \openone/d - (
\cdot )$, where $d$ denotes  the Hilbert space dimension, in
which case we introduce the short-hand notation $\overline{(
  \cdot )} := g ( \cdot )$.

The following lemma shows  that the probability distribution
attaining   the  maximum   in   the  right   hand  side   of
Eq.~\eqref{eq:gwmaximin}     and    the     numbering-valued
measurement attaining the minimum in  the right hand side of
Eq.~\eqref{eq:gwminimax}  share the  same symmetries  as the
c-q channel.

\begin{lmm}[Covariant case]
  \label{lmm:covariance}
  For any  finite set $\mathcal{M}$,  any finite-dimensional
  Hilbert  space   $\mathcal{H}$,  any   group  $\mathcal{G}
  \subseteq           \mathcal{S}_\mathcal{M}$,          any
  $\mathcal{G}$-covariant  c-q channel  $\boldsymbol{\sigma}
  \in  \mathcal{C}  (  \mathcal{M}, \mathcal{H})$,  and  any
  function $\gamma  : \{ 1, \dots  , | \mathcal{M} |  \} \to
  \mathbb{R}$,    there     exist    $\mathcal{G}$-invariant
  probability distribution $p \in \mathcal{D}_{\mathcal{M}}$
  and $\mathcal{G}$-covariant  measurement $\boldsymbol{\pi}
  \in \mathcal{P} (\mathcal{N}_{\mathcal{M}}, \mathcal{H} )$
  that  attain $G_{\max\min}^{\gamma}  ( \boldsymbol{\sigma}
  )$.
\end{lmm}

\begin{proof}
  Let us  prove the  first part of  the statement.   For any
  probability          distribution          $p          \in
  \mathcal{D}_{\mathcal{M}}$, upon  defining the probability
  distribution $q \in \mathcal{D}_{\mathcal{M}}$ given by $q
  := |  \mathcal{G} |^{-1} \sum_{g \in  \mathcal{G}} p \circ
  g$,  one has  $G_{\min}^{\gamma} (p,  \boldsymbol{\sigma})
  \le G_{\min}^{\gamma} (q,  \boldsymbol{\sigma})$. This can
  be seen as follows:
  \begin{align*}
    G_{\min}^{\gamma} \left(  p, \boldsymbol{\sigma} \right)
    &   =   \frac1{\left|\mathcal{G}\right|}   \sum_{g   \in
      \mathcal{G}}      G_{\min}^{\gamma}     \left(      p,
    \boldsymbol{\sigma}    \circ    g    \right)\\    &    =
    \frac1{\left|\mathcal{G}\right|}       \sum_{g       \in
      \mathcal{G}} G_{\min}^{\gamma} \left(  p \circ g^{-1},
    \boldsymbol{\sigma}         \right)\\        &         =
    \frac1{\left|\mathcal{G}\right|}       \sum_{g       \in
      \mathcal{G}}  G_{\min}^{\gamma}  \left(   p  \circ  g,
    \boldsymbol{\sigma}  \right)\\  & \le  G_{\min}^{\gamma}
    \left(  \frac1{\left| \mathcal{G}  \right|} \sum_{g  \in
      \mathcal{G}}    p    \circ   g,    \boldsymbol{\sigma}
    \right)\\    &    =    G_{\min}^{\gamma}    \left(    q,
    \boldsymbol{\sigma} \right).
  \end{align*}
  where the first  equality follows from the  fact that, due
  to  Eqs.~\eqref{eq:morphism}  and~\eqref{eq:equivariance},
  $G_{\min}^{\gamma} \left( p, \boldsymbol{\sigma} \right) =
  G_{\min}^{\gamma}  \left( p,  \boldsymbol{\sigma} \circ  g
  \right)$ for any $g  \in \mathcal{G}$, the second equality
  follows by  direct inspection, the third  equality follows
  from the group structure  of $\mathcal{G}$, the inequality
  follows    from   the    concavity   of    the   guesswork
  $G_{\min}^{\gamma}$       in      $p$       proven      in
  Lemma~\ref{lmm:convexity}, and the  final equality follows
  by definition of $q$.  By definition of $q$ it immediately
  follows  that  $q$  is   invariant  under  the  action  of
  $\mathcal{G}$, that  is, for  any $g \in  \mathcal{G}$ one
  has $q \circ g = q$.

  Let us  prove the second  part of the statement.   For any
  numbering-valued    measurement   $\boldsymbol{\pi}    \in
  \mathcal{P} (\mathcal{M}, \mathcal{H})$, upon defining the
  numbering-valued   measurement    $\boldsymbol{\tau}   \in
  \mathcal{P}    (\mathcal{M},   \mathcal{H})$    given   by
  $\boldsymbol{\tau}  (  \cdot  ) :=  |  \mathcal{G}  |^{-1}
  \sum_{g \in  \mathcal{G}} R_g^{-1}  \circ \boldsymbol{\pi}
  (g    \circ    \cdot)$,   one    has    $G_{\max}^{\gamma}
  (\boldsymbol{\sigma},         \boldsymbol{\pi})        \ge
  G_{\max}^{\gamma}          (          \boldsymbol{\sigma},
  \boldsymbol{\tau})$.

  This can be seen as follows:
  \begin{align*}
    G_{\max}^{\gamma}       \left(      \boldsymbol{\sigma},
    \boldsymbol{\pi}          \right)           &          =
    \frac1{\left|\mathcal{G}\right|}       \sum_{g       \in
      \mathcal{G}}  G_{\max}^{\gamma} \left(  R_g^{-1} \circ
    \boldsymbol{\sigma}     \circ    g,     \boldsymbol{\pi}
    \right)\\  & =  \frac1{\left|\mathcal{G}\right|} \sum_{g
      \in     \mathcal{G}}      G_{\max}^{\gamma}     \left(
    \boldsymbol{\sigma},  R_g \circ  \boldsymbol{\pi} \left(
    g^{-1}    \circ   \cdot    \right)    \right)\\   &    =
    \frac1{\left|\mathcal{G}\right|}       \sum_{g       \in
      \mathcal{G}}          G_{\max}^{\gamma}         \left(
    \boldsymbol{\sigma},  R_g^{-1}   \circ  \boldsymbol{\pi}
    \left(   g  \circ   \cdot   \right)   \right)\\  &   \ge
    G_{\max}^{\gamma}       \left(      \boldsymbol{\sigma},
    \frac1{\left|    \mathcal{G}   \right|}    \sum_{g   \in
      \mathcal{G}} R_g^{-1} \circ  \boldsymbol{\pi} \left( g
    \circ  \cdot  \right)  \right)\\ &  =  G_{\max}^{\gamma}
    \left( \boldsymbol{\sigma}, \boldsymbol{\tau} \right).
  \end{align*}
  where the first  equality follows from the  fact that, due
  to  Eqs.~\eqref{eq:morphism}  and~\eqref{eq:equivariance},
  $G_{\max}^{\gamma}       \left(       \boldsymbol{\sigma},
  \boldsymbol{\pi}   \right)   =  G_{\max}^{\gamma}   \left(
  R_g^{-1}  \boldsymbol{\sigma}  \circ  g,  \boldsymbol{\pi}
  \right)$ for any $g  \in \mathcal{G}$, the second equality
  follows by  direct inspection, the third  equality follows
  from the group structure  of $\mathcal{G}$, the inequality
  follows from the convexity  of the guesswork $G_{\max}$ in
  $\boldsymbol{\pi}$  proven  in  Lemma~\ref{lmm:convexity},
  and   the  final   equality  follows   by  definition   of
  $\boldsymbol{\tau}$.  By definition of $\boldsymbol{\tau}$
  it   immediately  follows   that  $\boldsymbol{\tau}$   is
  covariant under the action  of $\mathcal{G}$, that is, for
  any   $g    \in   \mathcal{G}$   one   has    $R_g   \circ
  \boldsymbol{\tau} = \boldsymbol{\tau} \circ g$.
\end{proof}

Let us consider  the case of a transitive  symmetry. In this
case, it is possible to  provide the closed-form solution to
the  maximization   problem  in  the  right   hand  side  of
Eq.~\eqref{eq:gwmaximin}  by noticing  that  the only  fully
invariant   probability   distribution    is   the   uniform
probability  distribution; it  is  not  possible however  in
general to solve in  closed-form the minimization problem in
the right hand side of Eq.~\eqref{eq:gwminimax} owing to the
operatorial structure of measurements.

Having shown that,  of the two optimization  problems we are
considering (maximization over probability distributions and
minimization over numbering-valued measurements), the former
is  ``easy'' and  the  latter is  ``hard''  (in the  precise
meanings discussed above), we focus  in the following on the
latter.

It will  be convenient  to restrict to  non--increasing cost
functions,  formalizing   the  expectation  that   the  cost
increases  with the  number of  queries needed  to correctly
guess.   That  this  restriction   comes  with  no  loss  of
generality  can be  shown as  follows.  For  any finite  set
$\mathcal{M}$  and  any  function  $\gamma  :  \{  1,  \dots
|\mathcal{M}|\}    \to    \mathbb{R}$,   let    us    define
$\overleftarrow{\gamma}  :=  \gamma   \circ  \sigma$,  where
$\sigma  :  \{1,  \dots  |\mathcal{M}|  \}  \to  \{1,  \dots
|\mathcal{M}|   \}$    is   any   permutation    such   that
$\overleftarrow{\gamma}$ is non  increasing.  It immediately
follows  (see Lemma~2  of Ref.~\cite{Dal23})  that, for  any
finite-dimensional  Hilbert  space  $\mathcal{H}$,  any  c-q
channel $\boldsymbol{\sigma} \in  \mathcal{C} ( \mathcal{M},
\mathcal{H} )$, and  any function $\gamma : \{ 1,  \dots , |
\mathcal{M} | \}$, one has
\begin{align*}
  G_{\min}^\gamma  \left( p,  \boldsymbol{\sigma} \right)  =
  G_{\min}^{\overleftarrow{\gamma}}         \left(        p,
  \boldsymbol{\sigma} \right).
\end{align*}
Moreover, if  numbering-valued measurement $\boldsymbol{\pi}
\in \mathcal{P}  ( \mathcal{M},  \mathcal{H} )$  attains the
minimum guesswork $G_{\min}^\gamma  ( p, \boldsymbol{\sigma}
)$,  then  numbering-valued measurement  $\boldsymbol{\pi}'(
\cdot  )  :=  \boldsymbol{\pi}  (\cdot  \circ  \sigma^{-1})$
attains           the            minimum           guesswork
$G_{\min}^{\overleftarrow{\gamma}} (  p, \boldsymbol{\sigma}
)$. Hence,  in the following  without loss of  generality we
may assume  whenever needed that the  cost function $\gamma$
is non decreasing.

The   following   lemma   (which  generalizes   Lemma~4   of
Ref.~\cite{DBK22}  to the  case of  arbitrary cost  function
$\gamma$)  formalizes through  Bayes  theorem the  intuition
that, for non decreasing cost function $\gamma$, the optimal
strategy for Bob  implies querying the values of  $M$ in the
order of their non increasing posterior probability.

\begin{lmm}[Bayes]
  \label{lmm:bayes}
  For  any discrete  set $\mathcal{M}$,  any non  decreasing
  cost function $\gamma  : \{ 1, \dots  |\mathcal{M}| \} \to
  \mathbb{R}$,   any   probability   distribution   $p   \in
  \mathcal{D}_{\mathcal{M}}$, any finite dimensional Hilbert
  space     $\mathcal{H}$,    and     any    c-q     channel
  $\boldsymbol{\sigma}    \in   \mathcal{C}    (\mathcal{M},
  \mathcal{H})$,   a   measurement   $\boldsymbol{\pi}   \in
  \mathcal{P}    (\mathcal{N}_{\mathcal{M}},   \mathcal{H})$
  minimizes        the       guesswork,        that       is
  $G_{\textrm{min}}^{\gamma}  ( p,  \boldsymbol{\sigma} )  =
  G^{\gamma} ( p,  \boldsymbol{\sigma}, \boldsymbol{\pi} )$,
  only    if     $p_{\boldsymbol{\rho},    \boldsymbol{\pi}}
  (\mathbf{n},   \cdot  )$   is  not   increasing  for   any
  $\mathbf{n} \in \mathcal{N}( \mathcal{M} )$.
\end{lmm}

\begin{proof}
  We  show  that,  for   any  numbering  valued  measurement
  $\boldsymbol{\pi}                                      \in
  \mathcal{P}(\mathcal{N}_{\mathcal{M}},      \mathcal{H})$,
  there    exists    a    numbering    valued    measurement
  $\boldsymbol{\pi}'                                     \in
  \mathcal{P}(\mathcal{N}_{\mathcal{M}},  \mathcal{H})$ such
  that      the     probability      distribution     $p_{p,
    \boldsymbol{\sigma},   \boldsymbol{\pi}'}   (\mathbf{n},
  \cdot)$ is  not increasing  for any  numbering $\mathbf{n}
  \in  \mathcal{N}_{\mathcal{M}}$   and  $G^{\gamma}   (  p,
  \boldsymbol{\sigma}, \boldsymbol{\pi}' )  \le G^{\gamma} (
  p, \boldsymbol{\sigma}, \boldsymbol{\pi} )$, with equality
  if    and    only     if    $p_{p,    \boldsymbol{\sigma},
    \boldsymbol{\pi}}    (\mathbf{n},    \cdot)   =    p_{p,
    \boldsymbol{\sigma},   \boldsymbol{\pi}'}   (\mathbf{n},
  \cdot)$    for     any    numbering     $\mathbf{n}    \in
  \mathcal{N}_{\mathcal{M}}$.  Let $\{ g_{\mathbf{n}} : \{1,
  \dots, | \mathcal{M} | \}  \to \{1, \dots, | \mathcal{M} |
  \}   \;   |    \;   g_{\mathbf{n}}   \textrm{   bijective}
  \}_{\mathbf{n}   \in   \mathcal{N}_{\mathcal{M}}}$  be   a
  numbering  indexed family  of permutations  such that  the
  probability   distribution   $p_{p,   \boldsymbol{\sigma},
    \boldsymbol{\pi}}  (  \mathbf{n}, g_{\mathbf{n}}(  \cdot
  ))$ is  not increasing  for any numbering  $\mathbf{n} \in
  \mathcal{N}_{   \mathcal{M}}$.    Let    function   $f   :
  \mathcal{N}_{\mathcal{M}}  \to  \mathcal{N}_{\mathcal{M}}$
  be given by the composition
  \begin{align*}
    f   \left(  \mathbf{n}   \right)  :=   \mathbf{n}  \circ
    g_{\mathbf{n}},
  \end{align*}
  for       any        numbering       $\mathbf{n}       \in
  \mathcal{N}_{\mathcal{M}}$.   Let   the  numbering  valued
  measurement     $\boldsymbol{\pi}'    \in     \mathcal{P}(
  \mathcal{N}_{\mathcal{M}}, \mathcal{H}  ) $ be  the coarse
  graining      of     numbering      valued     measurement
  $\boldsymbol{\pi}$ given by
  \begin{align*}
    \boldsymbol{\pi}'   \left(    \mathbf{n}'   \right)   :=
    \sum_{\mathbf{n}  \in f^{-1}\left[  \mathbf{n}' \right]}
    \boldsymbol{\pi} \left( \mathbf{n} \right),
  \end{align*}
  for       any       numbering       $\mathbf{n}'       \in
  \mathcal{N}_{\mathcal{M}}$, where $f^{-1} [ \mathbf{n}' ]$
  denotes the counter-image of $\mathbf{n}'$ with respect to
  $f$.  By direct computation, one has that
  \begin{align*}
    q_{\boldsymbol{\sigma},   \boldsymbol{\pi}'}  \left(   t
    \right)       &       =      \sum_{\mathbf{n}'       \in
      \mathcal{N}_{\mathcal{M}}} \sum_{\mathbf{n} \in f^{-1}
      \left[  \mathbf{n}'  \right]}   p  \left(  \mathbf{n}'
    \left( t \right)  \right) \Tr \left[ \boldsymbol{\sigma}
      \left(   \mathbf{n}'   \left(    t   \right)   \right)
      \boldsymbol{\pi}  \left(  \mathbf{n}  \right)  \right]
    \\ & = \sum_{\mathbf{n} \in \mathcal{N}_{\mathcal{M}}} p
    \left(  f \left(  \mathbf{n}  \right)  \left( t  \right)
    \right)   \Tr   \left[  \boldsymbol{\sigma}   \left(   f
      \left(\mathbf{n}  \right)  \left(  t  \right)  \right)
      \boldsymbol{\pi}  \left(  \mathbf{n}  \right)  \right]
    \\ &  = \sum_{\mathbf{n}  \in \mathcal{N}_{\mathcal{M}}}
    p_{p, \boldsymbol{\sigma},      \boldsymbol{\pi}}      \left(
    \mathbf{n}, g_{\mathbf{n}} \left( t \right) \right),
  \end{align*}
  for  any  $t   \in  \{1,  \dots,  |   \mathcal{M}  |  \}$.
  Therefore, by  construction one  has that  the probability
  distribution $q_{\boldsymbol{\rho},  \boldsymbol{\pi}}$ is
  majorized     by      the     probability     distribution
  $q_{\boldsymbol{\rho}, \boldsymbol{\pi}'}$, that is
  \begin{align*}
    \sum_{t    \in   \left\{    1,   \dots,    T   \right\}}
    q_{\boldsymbol{\rho}, \boldsymbol{\pi}}  (t) \le \sum_{t
      \in     \left\{     1,      \dots,     T     \right\}}
    q_{\boldsymbol{\rho},   \boldsymbol{\pi}'}    \left(   t
    \right).
  \end{align*}
  for any  $T \in  \{ 1,  \dots, |  \mathcal{M} |  \}$, with
  equality  if  and   only  if  $p_{p,  \boldsymbol{\sigma},
    \boldsymbol{\pi}}    (\mathbf{n},    \cdot)   =    p_{p,
    \boldsymbol{\sigma},   \boldsymbol{\pi}'}   (\mathbf{n},
  \cdot)$ for  any numbering  $\mathbf{n} \in  \mathcal{N} (
  \mathcal{M} )$.  Therefore, the statement follows.
\end{proof}

\subsection{The qubit case}
\label{sec:qubit}

Finally, in this section we  show that the maximization over
numbering--valued measurements can  be solved in closed-form
for  any  given  qubit   c-q  channel,  if  the  probability
distribution is  uniform and the cost  function is balanced.
For any  finite set  $\mathcal{M}$, any non  increasing cost
function  $\gamma :  \{  1,  \dots |  \mathcal{M}  | \}  \to
\mathbb{R}$  is  balanced  if   and  only  if  $\gamma(t)  +
\gamma(|\mathcal{M}| + 1 - t) = 2 \overline{\gamma}$ for any
$t \in \{1, \dots |\mathcal{M}\}$.

Among  balanced cost  functions,  it will  be convenient  to
restrict  to  those  whose   average  is  null.   That  this
restriction  comes without  loss  of generality  immediately
follows from the fact that
\begin{align*}
  G_{\min}^{\gamma} \left(  \left| \mathcal{M} \right|^{-1},
  \sigma \right)  = \overline{\gamma}  - G_{\min}^{\gamma_0}
  \left( \left| \mathcal{M} \right|^{-1}, \sigma \right),
\end{align*}
where  $\gamma_0 :=  \gamma  -  \overline{\gamma}$ has  null
average.

Before proceeding, we  need to define a  family of operators
that  will  allow  us   to  reframe  the  optimization  over
numbering-valued  measurements   as  a   quantum  hypothesis
testing problem.  For any  finite set $\mathcal{M}$, any non
decreasing balanced  cost function $\gamma  : \{ 1,  \dots |
\mathcal{M}|  \}  \to  \mathbb{R}$ with  null  average,  any
finite-dimensional  Hilbert  space  $\mathcal{H}$,  any  c-q
channel $\boldsymbol{\sigma} \in  \mathcal{C} ( \mathcal{M},
\mathcal{H}  )$,  let $E_{\boldsymbol{\sigma}}^{\gamma}$  be
the function given given by
\begin{align*}
  E_{\boldsymbol{\sigma}}^{\gamma} \left( \mathbf{n} \right)
  :=   \frac{2}{\left|   \mathcal{M}  \right|}   \sum_{t   =
    1}^{\left| \mathcal{M} \right|}  \gamma \left( t \right)
  \boldsymbol{\sigma}  \left(  \mathbf{n} \left(  t  \right)
  \right),
\end{align*}
for any numbering $\mathbf{n} \in \mathcal{N}_{\mathcal{M}}$.

We  are  now  in  a  position  to  introduce  a  branch  and
bound~\cite{LD60,  LMSK63,  Jen99}  (BB) algorithm  for  the
closed-form computation of the guesswork of qubit ensembles.
A branch and bound algorithm maximizes an objective function
over a  feasible set  by recursively splitting  the feasible
set into subsets, then  minimizing the objective function on
such subsets;  the splitting is called  branching.  For each
such subset, the  algorithm computes a bound  on the maximum
it is trying to find, and  uses such bounds to ``prune'' the
search space, eliminating the subsets that cannot contain an
optimal solution.

For  any $\mathcal{N}  \subseteq \mathcal{N}_{\mathcal{M}}$,
let us define
\begin{align*}
  t^*_{\mathcal{N}}     :=      \argmax_{\substack{t     \in
      \mathbb{N}\\\mathbf{n} \left( t  \right) = \mathbf{n}'
      \left( t  \right), \; \forall  \mathbf{n}, \mathbf{n}'
      \in \mathcal{N}}} t + 1,
\end{align*}
and
\begin{align*}
  \mathcal{M}^*_{\mathcal{N}}         :=         \mathcal{M}
  \setminus     \begin{cases}    \bigcup_{\mathbf{n}     \in
      \mathcal{N}}  \left\{  \mathbf{n}  \left(  t  \right),
    \mathbf{n} \left( \left| \mathcal{M} \right| - t \right)
    \right\}_{t  =  1}^{t^*_{\mathcal{N}}}   &  \textrm  {if
      $\boldsymbol{\sigma}$  CS},\\ \bigcup_{\mathbf{n}  \in
      \mathcal{N}}  \left\{  \mathbf{n}   \left(  t  \right)
    \right\}_{t    =   1}^{t^*_{\mathcal{N}}}    &   \textrm
            {otherwise}.
  \end{cases}
\end{align*}

\begin{dfn}[BB algorithm]
  \label{dfn:bb}
  For  any  finite  set   $\mathcal{M}$,  any  balanced  non
  increasing cost function $\gamma : \{ 1, \dots \mathcal{M}
  \} \to \mathbb{R}$ with  null average, any two-dimensional
  Hilbert  space   $\mathcal{H}$,  any   group  $\mathcal{G}
  \subseteq     \mathcal{S}_{\mathcal{M}}$,      and     any
  $\mathcal{G}$-covariant c-q channel $\boldsymbol{\sigma} (
  \mathcal{M},  \mathcal{H}   )$,  let  us  define   the  BB
  algorithm   given  by   the  objective   function  $\|   E
  _{\boldsymbol{\sigma}}^{\gamma}  \left(\cdot \right)  \|$,
  the   feasible  set   $\brn^j   (  \mathcal{N}^{\left(   0
    \right)}_{\mathcal{M}}  ) (  m )$  for arbitrary  $m \in
  \mathcal{M}$, where $j = 1$ if $\mathcal{G}$ is transitive
  and $j = 0$ otherwise and
  \begin{align*}
    \mathcal{N}^{\left(       0       \right)}_{\mathcal{M}}
    := \begin{cases} \left\{ \mathbf{n}  \in \mathcal{N}_{\mathcal{M}}
      \Big|   \sigma
      \left(  \mathbf{n}  \left(  \cdot  \right)  \right)  +
      \sigma   \left(  \overline{\mathbf{n}}   \left(  \cdot
      \right) \right)  = \openone \right\},  & \textrm{if
        $\boldsymbol{\sigma}$                             is
        CS,}\\          \mathcal{N}_{\mathcal{M}}          &
      \textrm{otherwise,}
    \end{cases}
  \end{align*}
  the branching rule
  \begin{align*}
    \brn(\mathcal{N})   :=  \left\{   \mathcal{N}_m  \subseteq
    \mathcal{N}    \Big|    \mathbf{n}    \in    \mathcal{N}_m
    \Leftrightarrow \mathbf{n}\left(t^*_{\mathcal{N}}\right) =
    m \right\}_{m \in \mathcal{M}^*_{ \mathcal{N} }},
  \end{align*}
  and the bounding rule
  \begin{align}
    \label{eq:bound}
    \bnd(\mathcal{N}) :=  \left(k + 1\right)  \left( \left\|
    \sum_{t  = 1}^{t^*_{\mathcal{N}}  - 1}  \gamma \left(  t
    \right) \boldsymbol{\sigma}  \left( \mathbf{n}  \left( t
    \right)     \right)     \right\|     +     \sum_{t     =
      t^*_{\mathcal{N}}}^{\frac{\left|  \mathcal{M} \right|}{2k}}
    \gamma \left( t \right) \right),
  \end{align}
  for arbitrary $\mathbf{n} \in  \mathcal{N}$, where $k = 1$
  if $\boldsymbol{\sigma}$ is CS and $k = 0$ otherwise.
\end{dfn}

\begin{thm}
  \label{thm:bb}
  For  any  finite  set   $\mathcal{M}$,  any  balanced  non
  increasing cost function $\gamma : \{ 1, \dots \mathcal{M}
  \}  \to  \mathbb{R}$,  any two-dimensional  Hilbert  space
  $\mathcal{H}$,    any    group   $\mathcal{G}    \subseteq
  \mathcal{S}_{\mathcal{M}}$,             and            any
  $\mathcal{G}$-covariant c-q channel $\boldsymbol{\sigma} (
  \mathcal{M},   \mathcal{H}  )$,   the   BB  algorithm   in
  Definition~\ref{dfn:bb}    computes     the    closed-form
  expression  of   the  guesswork  $G_{\min}^{\gamma}   (  |
  \mathcal{M}   |^{-1},   \boldsymbol{\sigma}  )$   on   $C$
  computing units in finite  time $T_C ( \boldsymbol{\sigma}
  )$ given by
  \begin{align*}
    T_C \left( \boldsymbol{\sigma} \right) \le \frac1C
    \begin{cases}
      \left(  \left| \mathcal{M}  \right|  -  1 \right)!   &
      \textrm{if  $\mathcal{G}$  is transitive,}  \\  \left|
      \mathcal{M} \right|!!   & \textrm{if  $\mathcal{G}$ is
        CS,}\\  \left(   \left|  \mathcal{M}  \right|   -  2
      \right)!! & \textrm{if $\mathcal{G}$ is transitive and
        CS,}\\     \left|     \mathcal{M}    \right|!      &
      \textrm{otherwise,}
    \end{cases}
  \end{align*}
  where $( \cdot )!!$ denotes the doubly factorial function.
\end{thm}
  
\begin{proof}
  Due to Theorem~1 and  Corollary~1 of Ref.~\cite{Dal23} one
  has
  \begin{align*}
    G_{\min}^{\gamma}     \left(      \left|     \mathcal{M}
    \right|^{-1},  \boldsymbol{\sigma} \right)  = -  \frac12
    \max_{\mathbf{n} \in  \mathcal{N}_{\mathcal{M}}} \left\|
    E_{p, \boldsymbol{\sigma}}^{\gamma}  \left( \mathbf{n}^*
    \right) \right\|.
  \end{align*}
  
  As  observed  in  Ref.~\cite{Dal23},  this  represents  an
  instance of the  quadratic assignment problem~\cite{Cel98}
  (QAP) that can be solved  by a BB algorithm.  Permutations
  can be generated as the leaves  of a tree, as shown by the
  following tree diagram  for the case $\mathcal{M}  = \{ a,
  b, c \}$.
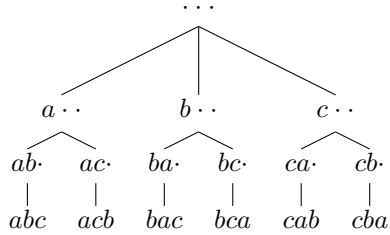
  \begin{align*}
    \Tree[.{$\cdot \cdot \cdot$} [.{$a \cdot \cdot$} [.{$a b
            \cdot$} [.{$a b c$} ]  ] [.{$a c \cdot$} [.{$a c
              b$} ] ] ] [.{$b  \cdot \cdot$} [.{$b a \cdot$}
          [.{$b a c$} ] ] [.{$b c  \cdot$} [.{$b c a$} ] ] ]
      [.{$c \cdot  \cdot$} [.{$c a  \cdot$} [.{$c a b$}  ] ]
        [.{$c b \cdot$} [.{$c b a$} ] ] ] ]
  \end{align*}
  Hence   the   complexity   for   the   general   case   is
  $|\mathcal{M}|!$.
  
  However,  in  the  presence  of  symmetries  this  can  be
  improved  upon  by  observing that  the  operators  $E_{p,
    \boldsymbol{\sigma}}^{\gamma}$'s inherit  the covariance
  of the c-q channel  $\boldsymbol{\sigma}$ under the action
  of the  statistical-morphism representation  of transitive
  group $\mathcal{G}$, that is
  \begin{align*}
    R_g  \circ E_{p,  \boldsymbol{\sigma}}^{\gamma} =  E_{p,
      \boldsymbol{\sigma}}^{\gamma}  \left(  g  \circ  \cdot
    \right).
  \end{align*}
  
  Hence,  if $\mathcal{G}$  is  transitive, for  any $m  \in
  \mathcal{M}$,   there   exists    an   optimal   numbering
  $\mathbf{n}^*$ such that  $\mathbf{n}^* ( 1 )  = m^*$, for
  any $m^*  \in \mathcal{M}$.  Permutations  satisfying this
  condition can be  generated as the leaves of  a tree whose
  root satisfies  the condition,  as shown by  the following
  tree diagram, this time for  the case $\mathcal{M} = \{ a,
  b, c, d\}$ and $m^* = a$.
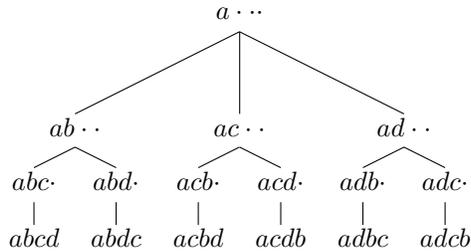
  \begin{align*}
    \Tree[.{$a  \cdot \cdot  \cdot$} [.{$a  b \cdot  \cdot$}
        [.{$a  b c  \cdot$} [.{$a  b  c d$}  ] ]  [.{$a b  d
            \cdot$}  [.{$a b  d  c$}  ] ]  ]  [.{$a c  \cdot
          \cdot$} [.{$a c b \cdot$} [.{$a  c b d$} ] ] [.{$a
            c d  \cdot$} [.{$a c d  b$} ] ] ]  [.{$a d \cdot
          \cdot$} [.{$a d b \cdot$} [.{$a  d b c$} ] ] [.{$a
            d c \cdot$} [.{$a d c b$} ] ] ] ]
  \end{align*}
  Hence  the complexity  in this  case is  $(|\mathcal{M}| -
  1)!$.

  Moreover,    if    $\mathcal{G}$    is    CS,    due    to
  Lemma~\ref{lmm:bayes}  there exists  an optimal  numbering
  $\mathbf{n}^*$ such that
  \begin{align*}
    \boldsymbol{\sigma}  \left(  \mathbf{n}^*  \left(  \cdot
    \right)    \right)    +    \boldsymbol{\sigma}    \left(
    \mathbf{n}^*  \circ  \sigma^{-1}_{\gamma}  \left(  \cdot
    \right) \right) = \openone.
  \end{align*}  
  Once again, permutations satisfying  this condition can be
  generated as  the leaves of a  tree if the $t$-th  and the
  $(|\mathcal{M}|  + 1  - t)$-th  states are  fixed at  each
  branch, as shown by the  following tree diagram, this time
  for the case $\mathcal{M} = \{ a, \bar{a}, b, \bar{b} \}$.
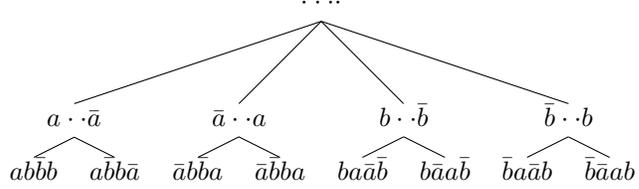
  \begin{align*}
    \Tree[.{$\cdot  \cdot \cdot  \cdot$}  [.{$a \cdot  \cdot
          \bar{a}$} [.{$a  b \bar{b}  b$} ] [.{$a  \bar{b} b
            \bar{a}$}  ]  ]   [.{$\bar{a}  \cdot  \cdot  a$}
        [.{$\bar{a} b  \bar{b} a$}  ] [.{$\bar{a}  \bar{b} b
            a$}  ] ]  [.{$b  \cdot \cdot  \bar{b}$} [.{$b  a
            \bar{a} \bar{b}$} ] [.{$b \bar{a} a \bar{b}$} ]]
      [.{$\bar{b} \cdot \cdot b$}  [.{$\bar{b} a \bar{a} b$}
        ] [.{$\bar{b} \bar{a} a b$} ] ] ]
  \end{align*}
  Hence the  complexity in  this case  is $|\mathcal{M}|!!$.
  
  Moreover, for a transitive and CS $\mathcal{G}$, combining
  the two results above the complexity becomes $|\mathcal{M}
  - 2|!!$,   attainable  by   generating  the   permutations
  according to the following  tree, again for $\mathcal{M} =
  \{ a, \bar{a}, b, \bar{b}\}$.
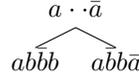
  \begin{align*}
    \Tree[.{$a \cdot  \cdot \bar{a}$} [.{$a b  \bar{b} b$} ]
      [.{$a \bar{b} b \bar{a}$} ] ]
  \end{align*}

  Finally, each  of the  trees above represents  a so-called
  embarrassingly parallel  problem, that is,  the computation
  in each branch  is independent of that  in other branches,
  so the  computation in different branches  can be assigned
  to  different computing  unit  with speedup  given by  the
  number of such units.
\end{proof}

A parallel  implementation in the C  programming language of
the BB  algorithm given  by Definition~\ref{dfn:bb}  is made
available under a free software license~\cite{BND23}.

We  are  now  in  a  position  to  compute  the  closed-form
expression of    the    guesswork   of    highly    symmetric,
informationally complete (HSIC) qubit c-q channels.

For  any finite  set $\mathcal{M}$  and any  two-dimensional
Hilbert     space    $\mathcal{H}$,     a    c-q     channel
$\boldsymbol{\sigma}   \in    \mathcal{C}   (   \mathcal{M},
\mathcal{H})$ is  highly-symmetric, informationally complete
(HSIC)~\cite{SS16} if  and only if  it is injective  and its
image $\boldsymbol{\sigma}  ( \mathcal{M} )$  corresponds in
the   Bloch   sphere   to  a   convex   regular   polyhedron
(tetrahedron,    octahedron,     cube,    icosahedron,    or
dodecahedron), the cuboctahedron, or the icosidodecahedron.

\begin{cor}[Guesswork of HSIC qubit c-q channels]
  \label{cor:hsic}
  For  any  finite  set $\mathcal{M}$,  any  two-dimensional
  Hilbert  space $\mathcal{H}$,  and  any  HSIC c-q  channel
  $\boldsymbol{\sigma}   \in   \mathcal{C}  (   \mathcal{M},
  \mathcal{H})$,     the    guesswork     $G_{\max\min}    (
  \boldsymbol{\rho} )$  for identity cost  function $\gamma(
  \cdot ) = \cdot$ is given by
  \begin{align*}
    G_{\max\min}    \left(    \boldsymbol{\sigma}    \right)
    = \begin{cases} \frac52 - \frac{\sqrt{15}}{6} \sim 1.9 &
      \left|   \mathcal{M}   \right|   =  4,\\   \frac72   -
      \frac{\sqrt{35}}6  \sim   2.5  &   \left|  \mathcal{M}
      \right| =  6,\\ \frac92 - \frac{\sqrt{7}}2  \sim 3.2 &
      \left|  \mathcal{M}   \right|  =  8,\\   \frac{13}2  -
      \frac{\sqrt{110 \left( 65  + 29 \sqrt{5} \right)}}{60}
      \sim    4.5   &    \left|   \mathcal{M}    \right|   =
      12,\\  \frac{21}2 -  \frac{\sqrt{6 \left(  3321 +  1483
          \sqrt{5}   \right)  }}{60}   \sim  7.2   &  \left|
      \mathcal{M} \right| = 20.
    \end{cases}
  \end{align*}
  if  $\boldsymbol{\sigma} (  \mathcal{M}  )$  is a  regular
  convex polyhedron,
  \begin{align*}
    G_{\max\min}    \left(    \boldsymbol{\sigma}    \right)
    = \frac{13}2 -   \frac{\sqrt{570}}6 \sim 4.5,
  \end{align*}
  if   $\boldsymbol{\sigma}   (    \mathcal{M}   )$   is   a
  cuboctahedron, and
  \begin{align*}
    G_{\max\min}   \left(   \boldsymbol{\sigma}  \right)   =
    \frac{31}2 -  \frac{\sqrt{117490 +  52534 \sqrt{5}}}{30
      \sqrt{6 + 2 \sqrt{5}}} \sim 10.5,
  \end{align*}
  if   $\boldsymbol{\sigma}   (   \mathcal{M}   )$   is   an
  icosidodecahedron.
\end{cor}

\begin{proof}
  Due to Lemma~\ref{lmm:covariance}  and the transitivity of
  $\mathcal{G}$, the probability  distribution $p$ attaining
  the maximin guesswork  $G_{\max\min} ( \boldsymbol{\sigma}
  )$  is uniform,  that is,  $p  \left( m  \right) =  \left|
  \mathcal{M}  \right|^{-1}$ for  any  $m \in  \mathcal{M}$.
  The guesswork  of regular convex polyhedra  was derived in
  Ref.~\cite{DBK22}. The guesswork  of the cuboctahedron was
  derived  in  Ref.~\cite{DBK21}.  However,  the  techniques
  derived therein do not suffice  for the computation of the
  guesswork of  the remaining  HSIC qubit c-q  channel, that
  is,  the  icosidodecahedral  c-q  channel  (for  which  $|
  \mathcal{M} |  = 30$), since the  practical application of
  such techniques is  limited to $| \mathcal{M}  | \sim 24$.
  We obtained such a result  by applying the BB algorithm in
  Definition~\ref{dfn:bb}, whose computational complexity in
  this  case is  $T_C(\boldsymbol{\sigma})  \le 28!!/C  \sim
  10^{15}/C$. We  initialized the feasible solution  using a
  greedy  algorithm (whose  complexity is  quadratic), whose
  output  turned out  to be  already within  $1\%$ from  the
  optimal  solution.  For  $C  = 16$,  the computation  took
  around one day.
\end{proof}

Finally, we implemented the  minimum guesswork setup for the
icosidodecahedral c-q channel  and run it on  an IBM quantum
computer.  We  generated one  circuit for  each of  the $30$
states of the ensemble, and ran $4000$ shots for each of the
$30$  circuits  using  the \verb|ibmq_quito|  backend.   The
resulting           minimum           guesswork           is
$G_{\max\min}(\boldsymbol{\rho}) \sim 10.4$, which is within
$1\%$     from      the     ideal      minimum     guesswork
$G_{\max\min}(\boldsymbol{\rho})  \sim   10.5$  reported  in
Corollary~\ref{cor:hsic}.

\section{Conclusion and outlooks}
\label{sec:conclusion}

In this  work we  addressed the  problem of  the adversarial
guesswork in  the presence of quantum  side information. For
the arbitrary-dimensional case, we  proved that i) the order
in which the strategies of  the two parties are optimized is
irrelevant,  ii) that  each  optimization  corresponds to  a
convex  problem,  and  that  iii) that  covariances  of  the
problem, if  any, are recast  as covariances in  the optimal
strategies.   We  conclusively  settled the  qubit  case  by
deriving  a   BB  algorithm  for  the   computation  of  the
closed-form expression  of the guesswork, and  we applied it
to compute the guesswork of HSIC c-q channels.

The problem of  the guesswork for symmetric  c-q channels in
arbitrary dimension remains open.  The difficulty stems from
the   fact  that,   while   the  operators   $E^{\gamma}_{p,
  \boldsymbol{\sigma}}$'s   preserve   the   symmetries   of
$\boldsymbol{\sigma}$,  a  (factorially   large)  number  of
inequivalent   seeds  is   needed   to   generate  all   the
$E^{\gamma}_{p,      \boldsymbol{\sigma}}$'s     from      a
representation   of   $\mathcal{G}$,   and   thus   standard
techniques for  semidefinite programming in the  presence of
symmetries cannot be applied directly.

All in  all, the guesswork  represents a relatively  new and
unexplored operational quantifier  of information.  As such,
it   has  promising   applications  in   contexts  such   as
information-disturbance    relations,    where   it    could
potentially be used in place of well-established quantifiers
of   information  such   as  the   mutual  information;   in
majorization theory, where operational  concepts such as the
testing  regions,  typically  defined   in  terms  of  error
probability, could  be redefined  in terms of  guesswork; in
quantum  cryptographic  applications,   where  bounding  the
guesswork  of  a  communication   channel  could  allow  for
alternative security proofs; in  witnessing the violation of
Bell inequalities; finally, the application of the guesswork
could  be extended  from c-q  channels to  encompass quantum
channels and quantum combs.

\section{Acknowledgments}

M.~D. is grateful to  Alessandro Bisio and Francesco Buscemi
for insightful discussions. M.~D.  acknowledges support from
the   Department  of   Computer  Science   and  Engineering,
Toyohashi University  of Technology, from  the International
Research Unit of Quantum  Information, Kyoto University, and
from the JSPS KAKENHI grant number JP20K03774.

\end{document}